\documentclass{article}

\usepackage{arxiv}

\usepackage[symbol]{footmisc}
\usepackage{graphicx}
\usepackage{enumitem}
\usepackage{booktabs}
\usepackage{amsmath}
\usepackage{amssymb}
\usepackage{amsthm}
\usepackage[linesnumbered,ruled,vlined]{algorithm2e}
\usepackage{threeparttable}
\newtheorem{proposition}{Proposition}

\DontPrintSemicolon%
\SetKwFunction{Func}{Function}%
\SetKwComment{tcc}{/*}{*/}%
\SetKwInOut{Input}{input}\SetKwInOut{Output}{output}%
\SetKwData{false}{false}\SetKwData{true}{true}
\SetKwData{done}{done}
\SetKwData{bestScore}{bestScore}
\SetKwData{trialScore}{trialScore}
\SetKwData{buildNode}{buildNode}
\SetKwData{op}{op}
\SetKwData{Value}{value}
\SetKwData{newNode}{newNode}
\SetKwData{nodeType}{nodeType}
\SetKwData{maxLevel}{maxLevel}
\SetKwData{level}{level}
\SetKwData{leaf}{leaf}
\SetKwData{left}{left}
\SetKwData{right}{right}
\SetKwFunction{BuildNode}{BuildNode}

\usepackage{cite}

\usepackage{tikz}
\usetikzlibrary{matrix}
\usetikzlibrary{calc, patterns}
\usepackage{pgfplots}
\pgfplotsset{compat=1.15}

\definecolor{citeblue}{rgb}{0.1,0,.4}
\usepackage{hyperref}
\definecolor{darkblue}{rgb}{0,0,0.6}
\definecolor{darkred}{rgb}{0.7,0,0}

\begin{document}

\title{Boosting Isomorphic Model Filtering with Invariants%\thanks{Grants or other notes
%about the article that should go on the front page should be
%placed here. General acknowledgments should be placed at the end of the article.}
}
% \subtitle{Do you have a subtitle?\\ If so, write it here}

%\titlerunning{Short form of title}        % if too long for running head

\author{
   \href{https://orcid.org/0000-0001-6655-2172}{\includegraphics[scale=0.06]{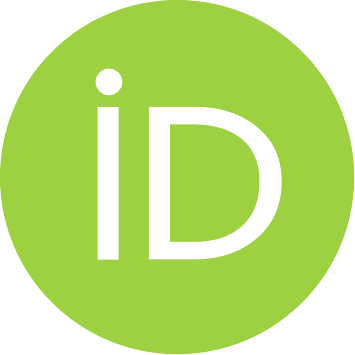}\hspace{1mm}}%
   Jo\~ao Ara\'ujo\\
    Universidade Nova de Lisboa, Lisbon, Portugal \\
    \texttt{jj.araujo@fct.unl.pt}           %  \\
	 \And
   \href{https://orcid.org/0000-0002-2067-0568}{\includegraphics[scale=0.06]{orcid.pdf}\hspace{1mm}}%
   Choiwah Chow\thanks{Corresponding author} \\
   Universidade Aberta, Lisbon, Portugal \\
   \texttt{1702603@estudante.uab.pt}    
	 \And
   \href{https://orcid.org/0000-0003-3487-784X}{\includegraphics[scale=0.06]{orcid.pdf}\hspace{1mm}}%
        Mikol\'a\v{s} Janota \\
        Czech Technical University in Prague, Czechia \\
        \texttt{mikolas.janota@cvut.cz}    
}

\maketitle

\begin{abstract}
The enumeration of finite models %of first order logic formulas is an indispensable
is very important to the working discrete mathematician  (algebra, graph theory, etc) and hence the search for effective methods to do this task is a critical goal  
 in discrete computational mathematics. However, it is hindered by the possible existence of many isomorphic models, which
usually only add noise. Typically, they are filtered out {\em a posteriori}, a step that might take a long time just to discard redundant models.  
This paper proposes a novel approach to split the generated models
into mutually non-isomorphic blocks. To do that we use well-designed hand-crafted invariants
as well as randomly generated invariants. The blocks are then tackled separately and possibly in parallel.
This approach is integrated into Mace4 (the most popular tool among mathematicians) where it shows tremendous 
speed-ups for a large variety of algebraic structures.

% \PACS{PACS code1 \and PACS code2 \and more}
% \subclass{08-08}% MSC codes
\end{abstract}

%%%%%%%%%%%%%%%%%%%%%%%%%%%%%%%
\section{Introduction}%
\label{sec:introduction} 
To study and get intuition on different types of relational algebras (groups, semigroups, and their ordered versions, quasigroups, fields, rings, MV-algebras, lattices, etc.), mathematicians resort to libraries of all order $n$ models (for small values of $n$) of the algebra they are interested in. These libraries allow experimentation such as forming and/or testing conjectures etc., to gain insights into the algebras in question. Indeed, GAP~\cite{GAP4}, the most popular computational algebra system, has many libraries of different algebras (semigroups, quasigroups, etc.), and more are needed. These libraries are so important that the search for them has a long history in mathematics predating for many years the use of computers. (See Appendix B of~\cite{bookDixonMortimer_1996PermutationGroup}; and for more recent results please see OEIS~\cite{oeis}, where many of the sequences are the number of order $n$ non-isomorphic models of a given class of algebras.)

Many of these algebras can be defined in first-order logic (FOL) and there are tools to allow mathematicians to encode their algebras and produce a meaningful library. The problem is that usually the tools, such as Mace4~\cite{Mccune_2003_tech264}, which can be easily learned and used by mathematicians and hence is very popular among them, generate too many isomorphic models (see Section~\ref{sec:mathbackground} for the definition of isomorphism) that need to be eliminated~\cite{DBLP:journals/jar/AudemardBH06}. 
For example, out of the 230,984 models for the implication algebras (see definition in~\cite{alfAraujoMatosRamires}) of order 10 generated by Mace4,
only 18 ($\approx 0.0078\%$) of them are pairwise non-isomorphic.

Redundant models may either be eliminated during the search phase or filtered out afterwards.
Guaranteeing that search never produces isomorphic models is a hard problem and is rarely seen in modern solvers.
This paper, therefore, tackles the second problem, i.e.,
the removal of redundant models from an already enumerated set.

In the context of finite model enumeration, the complexity of checking whether two models are isomorphic is
only part of the problem.  Another source of complexity is the large number of
models that need to be checked.  If every model is checked against all others, then the
performance degrades rapidly as the total number of models increases.

If we assign to each generated model a vector that is invariant under isomorphism and put all models having the same invariant vectors into separate blocks, then models across the blocks will not be isomorphic. This splits the problem into substantially smaller sub-problems. Moreover, processing of the blocks can easily be done in parallel as models across blocks cannot be isomorphic.
Parallel processing is an important facet of our approach since
modern-day computers are more often than not equipped with multiple cores.

Our contributions to the area of isomorphic model elimination are\footnotemark[1]:
\begin{enumerate}
	\item Devise an invariant-based parallel algorithm that can be applied to algebras defined by first-order logic formulas and containing at least one binary operation or relation (see Section~\ref{sec:algorithm}).
	\item Design a small basic set of invariants that have high discriminating power, and yet are inexpensive to compute (see Section~\ref{sec:basic:invariants}).
	\item Add randomly generated invariants to the invariant-based algorithm to help discover invariants of high discriminating power (see Section~\ref{sec:randominvariants}).
	\item Use a hash-map to store models partitioned by the invariant-based algorithm to allow fast storage and retrieval of models in the same block (see Section~\ref{sec:algorithm}). 
\end{enumerate}

Our goal is to help mathematicians on two levels: first, provide them with a tool on their desktop that quickly produces a library for the algebra they are working on; second, run the tool on a cluster of computers to precompute libraries for the most famous classes of algebras, and add them to GAP~\cite{GAP4} or a similar system.

\footnotetext[1]{Some preliminary ideas and results have been presented in~\cite{araujo_et_al:LIPIcs.CP.2021.4}. This paper adds more invariants including randomly generated invariants and proves their validity. It also reports substantially more experimental results and drills deeper into related work.}

%%%%%%%%%%%%%%%%%%%%%%%%%%%%%%%

\section{Definitions and Preliminaries}
\label{sec:mathbackground}
We give a brief overview of the mathematics used in the subsequent sections;
we draw mainly from Chapter 2 of~\cite{bookBurrisSankappanavr_2000UniversalAlgebra}.

A relational algebra is a triple $(D, \Sigma_F, \Sigma_R)$,
where $D$ is a set and $\Sigma_F$ is a set of operations, that is, functions $f:\,D^n \rightarrow D$
and $\Sigma_R$  is a set of relations, i.e., $R\in\Sigma_R$ is a subset of $D^n$.
%(in this case $f$ is said to be an operation of arity $n$).
%
The \textit{order} of a relational algebra is the size of its domain $D$. (Recall that examples of relational algebras are all imaginable algebras, (di)graphs, etc.; in the following, by algebra we mean relational algebra.)  

While  the concept of isomorphism is ubiquitous to scientific literature, its definition appears under slight variations.
Throughout this paper, we rely on the following definition.
Let $A$ and $B$  be structures defined on the same signature $\Sigma_F, \Sigma_R$.
A function $f$ from $A$ to $B$ is said to be an \emph{isomorphism} if it is a bijection and preserves all operations and relations. This means that if $g\in\Sigma_F$,  with the respective interpretations $g^A$ and $g^B$ in $A$ and $B$,
then $f(g^A(a_1,\dots,a_n))=g^B(f(a_1),\dots,f(a_n))$, for all $a_1,\dots,a_n\in A$.
Analogously, $f$ preserves $R\in\Sigma_R$ of arity $n$, with the respective interpretations $R^A$ and $R^B$ in $A$ and $B$, if $(a_1,\ldots, a_n)\in R^A$ implies 
$(f(a_1),\dots,f(a_n))\in R^B$, for all $a_1,\ldots,a_n\in A$.

% \input{iso_ja}

%If two models are isomorphic, then they are elementarily equivalent 
%(\textit{cf.}~Theorem 1.1.10 in~\cite{Marker2003-MARMTA-14}). That is, first-order formulas are \textit{invariant} under, or in other words, preserved by, all isomorphisms between two finite models. 

An important property of isomorphisms is that they preserve sets defined by some fixed formula.
More precisely, suppose we have two finite relational algebras $A$ and $B$, on a signature $\Sigma$, isomorphic under $f:A\to B$.
In addition, suppose we have a set $S$ contained in $A^k$ and definable by a FOL
formula $\Phi$ in the language of $\Sigma$. Then $f(S)$ is precisely the subset of $B^k$ that satisfies $\Phi$
(\textit{cf.}~Theorem 1.1.10 in~\cite{Marker2003-MARMTA-14}).

For example, suppose we have two isomorphic finite algebras:  $(A,*)$ and
$(B,*)$, with $f:A\to B$ an isomorphism. Suppose also that $S$ is the set $\{
  (x,y) \in A \mid (\exists z \in A) \ (x *^A z) *^A y = x *^A (z  *^A y)\}$. As $S$ is
the set of all elements in $A^2$ that satisfy a FOL in a language with the function symbol~$*$,
then $f(S):=\{(f(x),f(y))\mid (x,y)\in S\}$  is precisely the set of all pairs  $(x,y)\in B^2$ such that
$(\exists z \in B) \ (x *^B z) *^B y = x *^B (z  *^B y)$.

This idea is usually expressed by saying that sets definable by FOL formulas
are invariant (or preserved) under isomorphism. This guarantees that when we
split the list of algebras using invariants  based on defining formulas, algebras in different blocks are
non-isomorphic; algebras inside the same block might be isomorphic or non-isomorphic. 
Therefore, to discard the redundant algebras we only have to
check within each block. This is the ground for our invariants-based algorithm.
For future reference, we state these considerations as a proposition.

%  is preserved by  \textit{invariant} under i all isomorphisms between two finite models. 

\begin{proposition}\label{prop:solution:FOLformula}
Let $A$ and $B$ be algebras of a signature $\Sigma$ and $f:A\to B$ be  their
isomorphism. Then any vector $(a_1, a_2, \dots, a_m) \in{A^m}$ satisfies a
first-order formula in the common language $\Sigma$ if
and only if the vector $(f(a_1), f(a_2), \dots, f(a_m))\in{B^m}$ satisfies
the same first-order formula in $B$.
\end{proposition}
%\begin{proof}
%	This follows directly from the fact that isomorphic models are elementarily equivalent.\qed
%\end{proof}
%
%

%%%%%%%%%%%%%%%%%%%%%%%%%%%%%%%

\section{Basic Invariants}%
\label{sec:basic:invariants}

%It should be obvious that if 2 models do not satisfy the conditions laid out in Proposition~\ref{prop:solution:FOLformula} cannot be isomorphic. 
%This is the basis of the invariant-based algorithm which aims to efficiently separate non-isomorphic models.

The goal of this section is to introduce our list of basic invariants. 

We start by observing that the axioms satisfied by an algebra might render some invariants useless.  For example, if the algebra is a group, then the invariant that counts the number of idempotents would be useless (there is exactly one idempotent in each group).   Thus, we need multiple invariants with inner algebraic meaning and large discriminating power in order to  target as many different algebraic properties/classes as possible. On the other hand, we should choose properties that are inexpensive to compute and not very many as that could slow down the computation. 

Our choices of invariants are based on concepts ubiquitous in various kinds of algebras. For example, one of our invariants (\ref{inv:idempotent}) is based on the fact that idempotents appear in many algebras; in particular, it is well-known that every finite semigroup has at least one idempotent and hence this invariant is useful for a wide range of algebras, especially those that have a semigroup reduct.

As observed above, the overwhelming majority of the most popular algebras are defined using operations of arity at most 2 (see page 26 of~\cite{bookBurrisSankappanavr_2000UniversalAlgebra}), so we design most of our invariants around binary operations. We have 10 invariants from binary operations, 4 from binary relations, 4 from unary operations, and 1 from ternary operations to target different common algebraic structures.  Together they have high discriminating powers, and yet are easy and inexpensive to compute.

In the following discussions on invariants, the domain of the algebra is denoted by $D=\{1,2,\cdots,n\}$.

\subsection{Invariants from Unary Operations}
\label{subsec:unary_invariants}
For each unary operation $g$ in the algebra, we compute the invariants for each element $x\in{D}$:
\begin{enumerate}[label=\textbf{U\arabic*}]
	\item 1 \label{unary:inv:1} if $g(x) = x$, 0 otherwise (fixed point);
	\item 1 \label{unary:inv:2} if $g(x) \neq {x}$ and $g(g(x)) = x$, 0 otherwise (transposition);
	\item \label{unary:inv:3} The number of $y\in{D}$ such that $g(y)=x$ (size of the inverse image);
	\item \label{unary:inv:4} The number of $y\in{D}$ such that $g(g(y))=x$ (size of the inverse image under $g^2$).
\end{enumerate}

The correctness of these invariants 
%of~\ref{unary:inv:1} and~\ref{unary:inv:2} as invariants 
follows readily from Proposition~\ref{prop:solution:FOLformula}.
% The correctness \ref{unary:inv:3} and \ref{unary:inv:4} also follows from Proposition~\ref{prop:solution:FOLformula} noting that isomorphism is a bijection.
Of course, the correctness of these invariants can also be proved directly (without using Proposition~\ref{prop:solution:FOLformula}); as a sample illustration, we provide the details for invariant~\ref{unary:inv:3} in the next result. 

\begin{proposition}\label{prop:u3}
	Let $f$ be an isomorphism of algebras $A$ and $B$
	and $v_A, v_B$  be invariant vectors calculated according to~\ref{unary:inv:3} for $A$ and $B$, respectively.
	Then $v_A=v_B$.
\end{proposition}
\begin{proof}
	For $C\in \{A,B\}$, $x\in C$,
	and $g^C$ an interpretation of $g$ in $C$,
	define $K^C_x:=\{y \in C \mid g^C(y)=x\}$. We claim that $K^B_{f(x)}=f(K^A_{x})$,
	i.e.\ $x$ and $f(x)$ will be represented by the same number $\lvert{K}^B_{f(x)}\rvert=\lvert{f}(K^A_{x})\rvert$,
	in the (sorted) vectors $v_A$ and $v_B$,  respectively.

	If $y \in K^A_{x}$, then $g^A(y)=x$ so that $f(g^A(y))=f(x)$ and hence $g^B(f(y))=f(x)$, that is, $f(K^A_{x}) \subseteq K^B_{f(x)}$. It follows that 
	$$\lvert{K}^A_{x}\rvert=\lvert{f}(K^A_{x})\rvert \leq \lvert{K}^B_{f(x)}\rvert.$$

	As $f^{-1}:B\to A$ is an isomorphism too, the foregoing argument shows that 
	$$\lvert{K}^B_{u}\rvert=\lvert{f}^{-1}(K^B_{u})\rvert \leq \lvert{K}^A_{f^{-1}(u)}\rvert.$$
	Taking in the previous formula $u=f(x)$ we get 
	$$\lvert{K}^B_{f(x)}\rvert=\lvert{f}^{-1}(K^B_{f(x)})\rvert \leq \lvert{K}^A_{f^{-1}(f(x))}\rvert=\lvert{K}^A_{x}\rvert.$$

	Now $\lvert{K}^A_{x}\rvert=\lvert{f}(K^A_{x})\rvert \leq \lvert{K}^B_{f(x)}\rvert\leq \lvert{K}^A_{x}\rvert$
	implies $\lvert{K}^A_{x}\rvert=\lvert{f}(K^A_{x})\rvert = \lvert{K}^B_{f(x)}\rvert$. As $f(K^A_{x}) \subseteq K^B_{f(x)}$ it follows that $f(K^A_{x}) = K^B_{f(x)}$ (as we are dealing with finite structures).

%
%In fact, $y \in K^B_{f(x)}$ implies that $y\in B$ and $g(y)=f(x)$. As $f$ is surjective, there exists a $z\in A$ such that $f(z)=y$. Therefore, $g(f(z))=f(x)$ and hence $f(g(z))=f(x)$ so that ($f$ is injective) $g(z)=x$, that is, $z\in K^A_x$.   
%	Consider all the satisfying pairs $(d_1, d_2)$ of $f(y)=x$ in $A$ and in $B$.
%	Let us denote these two sets as $S_A$ and $S_B$.
%	For an element $d\in A$ consider the set $S_A^d=\{(d_1, d_2)\in S_A\mid d_1=d\}$.
%	Analogously, consider the set $S_B^{f(d)}=\{(d_1, d_2)\in S_B\mid d_1={f(d)}\}$.
%%
%	From Proposition~\ref{prop:solution:FOLformula}, $f$ is a bijection between
%	$S_A^{d}$ and $S_B^{f(d)}$,  so that
%	$(d_1,d_2)\in S_A^{d}$ iff $(f(d_1),f(d_2))\in S_B^{f(d)}$.
%	This means that $|S_A^{d}|=|S_B^{f(d)}|$,
%	i.e., the values of~\ref{unary:inv:3}  for the element $x$ and $f(x)$ are the same
%	in $v_A$ and $v_B$, respectively.
%
%	Since the invariant vectors $v_A$ and $v_B$ comprise the sorted values of~\ref{unary:inv:3} 
%	for element $d\in A$ and $d\in B$, respectively, they must necessarily be equal.
\end{proof}

Similar arguments can easily be used to directly prove the correctness of the other invariants.
  
%For brevity, we shall omit the proofs of correctness of the invariants in the following sections as they all follow the same arguments in Proposition~\ref{prop:u3}.

\subsection{Invariants from Binary Operations}
\label{subsec:binary_invariants}
%Binary operations are the most prevalent operations in algebras.  They have a lot of structures in them, so they provide fertile grounds for harvesting the most powerful invariants. 

For each domain element $x\in{D}$, we compute the following invariants for each binary operation in the algebra:
\begin{enumerate}[label=\textbf{B\arabic*}]
	\item The smallest integer $n$ such that $x^n = x^k, n > k >= 1$ where we define $x^n$ to be $(\ldots(x*x)*x)*x)\ldots$ for $n$ $x$'s (\emph{periodicity}).
	\item The number of $y\in D$ such that $x = (xy)x$ (\emph{number of inverses}).
	\item The number of distinct $xy$ for all $y\in{D}$ (\emph{size of right ideal}).
	\item The number of distinct $yx$ for all $y\in{D}$ (\emph{size of left ideal}).
	\item \label{inv:idempotent} 1 if $xx=x$, 0 otherwise (\emph{idempotency}).
	\item The number of $y\in D$ such that $x(yy) = (yy)x$ (\emph{number of commuting squares}).
	\item The number of $y\in D$ such that $x=yy$  (\emph{number of square roots}).
	\item The number of $y\in D$ such that $x(xy) = (xx)y$ (\emph{number of square associatizers}).
	\item The number of pairs of $y, z\in{D}$ such that $zy = yz = x$ (\emph{number of commuting pairs}).
	\item The number of $y\in{D}$ such that there exist pairs of $s, t\in{D}$ where $x=st$ and $y=ts$ (\emph{number of conjugates}).
\end{enumerate}

\subsection{Invariants from Binary Relations}
\label{subsec:binary_relations_invariants}
%Binary relations are not as common as binary operations in algebras, and often have less structures in them. So, we calculate fewer invariants for them. 

For each domain element $x\in{D}$, the following invariants are calculated for each binary relation $R$:
\begin{enumerate}[label=\textbf{R\arabic*}]
	\item The number of distinct $y$ such that $R(x, y)$.
	\item The number of distinct $y$ such that $R(y, x)$.
	\item 1 if $R(x, x)$, 0 otherwise. (reflexivity)
	\item The number of $y$ such that $R(x, y)\ \&\  R(y, x)$.
\end{enumerate}

\subsection{Invariants from Ternary Operations}
\label{subsec:ternary_operation_invariants}
Ternary operations are very rare. Indeed, no ternary operation exists in the definition in any of the 158 algebras listed in the ALF database~\cite{alfAraujoMatosRamires}, although a few of them  come from the Skolemization of binary operations. Moreover, calculations involving ternary operations are often very expensive as deeply nested loops are involved. Thus, only one simple invariant would be included in the algorithm. For each domain element $x\in{D}$, we compute one invariant for each ternary operation:
\begin{enumerate}[label=\textbf{T\arabic*}]
	\item The number of times $x$ appears in the ternary operation table. (frequency)
\end{enumerate}

We call the hand-crafted invariants listed above the \emph{basic invariants} to differentiate them from the randomly generated invariants which will be discussed in Section~\ref{sec:randominvariants}. It should be simple to see that the validity of these basic invariants follows from Proposition~\ref{prop:random:invariant}.

%%%%%%%%%%%%%%%%%%%%%%%%%%%%%%%

\section{Random Invariants}%
\label{sec:randominvariants}
As discussed in Section~\ref{sec:basic:invariants}, we need different invariants to target different algebraic structures. % and hopefully to increase the overall discriminating power of the basic invariants.
 The basic invariants are inspired by our knowledge of the most popular algebras, however, there are many other (less common) algebraic structures and new ones are permanently appearing.  % to identify and to write specific invariants for.  
%The problem is further complicated by the interactions between multiple operations in the same algebra.  It is particularly difficult to decide which of the many possible combinations of unary and binary operations 
%in a particular algebra offer good discriminating power under different conditions.

Therefore, we need a general way (adaptable to each  class of algebras)  of generating invariants with good discriminating power.  A practical solution to this problem is to generate a large set of invariants with a random number of operations and a random number of variables, and then automatically discover the best subset to use.

\subsection{Generation of Random Invariants}%
\label{subsec:generationOfRandomInvariants}

A first-order formula can be represented by an expression tree with operations in the internal nodes and variables in the leaves. For example, Fig.~\ref{expression:tree} shows the tree representation of the first-order formula 
\begin{equation}\label{eq:random:invariants}
x = (x*y)+x'
\end{equation}
It is simple to randomly generate such an expression tree, as shown in Algorithm~\ref{algorithm:generate:random:invariant}. For simplicity, the root is set to be an operation that evaluates to true or false, and it always has two children. It is the only node that can be assigned the equality relation.  It may also be a binary relation if there is one in the algebra.

\begin{figure}
	\centering
	\begin{tikzpicture}[minimum size=.55cm,
	evel 2/.style={sibling distance=1.5cm},level 3/.style={sibling distance=.75cm},
	yscale=0.6,xscale=2,
	every node/.style={draw=gray,fill=blue!3}]
\node[circle,draw](z){$=$}
child{%
	node[circle,draw]{$x$}}
child{%
	node[circle,draw]{$+$} 
	child{node[circle,draw] {$*$}
		  child{node[circle,draw] {$x$}}
		  child{node[circle,draw] {$y$}}} 
	child{node[circle,draw] {$'$} child{node[circle,draw]{$x$}}} 
};
\end{tikzpicture}
	\label{expression:tree}
	\caption{Expression tree for $x = (x*y)+x'$.}
\end{figure}
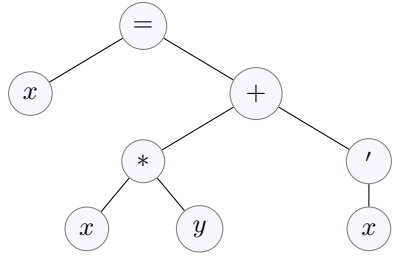

\begin{algorithm}[ht]
	\caption {Generation of Random First-order Formula}
	\label{algorithm:generate:random:invariant}
	\Input{A list of binary/unary operations/relations, max depth of tree, list of variables}
	\Output{An expression tree representing a first-order formula}\;
	
	\Func \BuildNode(\level)
	\Begin{%
		\tcc{recursively build a node}
		\tcc{\maxLevel is maximum depth of tree allowed}
		\If{\level = \maxLevel}{\nodeType $\gets$ \leaf\;}
		\Else{\nodeType $\gets$ random pick a \leaf, unary, or binary operation\;}
		create \newNode\;
		\If{\nodeType = \leaf}{%
			\newNode.\Value $\gets$ randomly pick a variable\;
			\Return \newNode\;}
		\newNode.\op $\gets$ randomly pick a unary or binary operation\;
		\newNode.\left $\gets$ \buildNode(\level+1)\;
		\If{\newNode.\op is a binary operation} { 
			\newNode.\right $\gets$ \buildNode(\level+1)\; }
		\Return \newNode
	}

	\Begin{%		
		create root node $R$\;
		$R.\op$ $\gets$ randomly pick the equality operation, or one of the binary relations (if exists)\;
		$R$.\left $\gets$ \BuildNode(1)\;
		$R$.\right $\gets$ \BuildNode(1)\;
		\Return $R$ \;
	}
\end{algorithm}

\begin{proposition}\label{prop:random:invariant}
	If we fix a variable in the first-order formula generated by
	Algorithm~\ref{algorithm:generate:random:invariant} as the base variable,
	then the number of solutions will be an invariant.
	% then the number of solutions will be an invariant for that variable.
\end{proposition}
\begin{proof}
	% A first order formula with one variable replaced by a constant is still a first order formula, so Proposition~\ref{prop:solution:FOLformula} applies.
	Correctness is shown as in Proposition~\ref{prop:u3}.\qed
\end{proof}

We may choose any of the variables in the randomly generated first-order formula as the base variable for the invariant. For example, if we choose $x$ as the base variable for the first-order formula~(\ref{eq:random:invariants}), then the invariant would read: The number of $y\in{D}$ such that $x = (x*y)+x'$, which is a valid invariant by  the foregoing argument.
%Corollary~\ref{prop:solution:count}.

\subsection{Quality Measure of Random Invariants}%
\label{subsec:qualityMeasureOfRandomInvariants}
After a large set of invariants is generated, a small subset will be selected based on its ability to reduce the work of the next step in filtering out isomorphic models. For a block with $m$ models, the worst-case scenario requires comparing every pair of models for isomorphism.  There would be $m(m-1)/2$ comparisons in total.  Based on this observation, we measure the quality of invariants by a score as follows: For a collection of invariants that induces a set of blocks of models \{$S_i\mid 1\leq{i}\leq{n}$\}, its score is computed as
\begin{equation}\label{eq:1}
\sum_{i\in \{1..n\}} |S_j|(|S_j|-1)
\end{equation}
The goal is to find the set of invariants having the minimum score over all possible combinations of randomly generated invariants in conjunction with the basic invariants.

\subsection{Selecting Random Invariants}%
\label{subsec:optimal:random:invariants}

We are not aware of any tractable algorithm for finding the optimal subset from a large set of random invariants according to the quality measure stated above.%
\footnote{We conjecture that the problem is NP-hard; it resembles K-means clustering, which is NP-hard~\cite{AloiseML09}.}
A feasible solution is to apply a greedy algorithm (see page 282 of~\cite{DBLP:books/ph/PapadimitriouS82} for the general discussions) to a small sample of the models to find an approximate optimal subset. In practice, it is sufficient to use a sample size of 0.1--0.2\%, or a thousand, whichever is larger, of the original set of models (see Section~\ref{sec:experiments} for discussions).  The algorithm is detailed in Algorithm~\ref{algorithm:select:random:invariant}.  The idea of the algorithm is to start with the basic invariants, then add the random invariants and calculate the scores one-by-one, keeping the random invariant only if it gives a better score.  Then repeat the process of adding random invariants, calculating the scores, and picking the best random invariant that minimizes the score until it cannot be further improved, or a preset maximum number of trials is reached.
This subset of random invariants, which may or may not be truly optimal, will then be used together with the basic invariants for the next step.

\begin{algorithm}[ht]
	\caption{Selecting Random Invariants with Greedy Algorithm}
	\label{algorithm:select:random:invariant}
	\Input{A set of random invariants $R$, a set of models $M$, and maximum trials $T$}
	\Output{A set $K\subseteq R$ of random invariants with $|K|<T$}
	
  	% \Begin{%
	$K \gets \emptyset$ \;
	$\bestScore\gets \infty$\;
		$\done\gets\false$\;
		\While{$\neg{\done} \land |K|<T$}{%
			$\done\gets\true$\;
			$a\gets\emptyset$ \;
			\ForEach{$r \in R\smallsetminus K$}{%
				 $\trialScore\gets$ score of $K\cup\{r\}$ on $M$ according to equation~\eqref{eq:1}\;
				\If{$\trialScore < \bestScore$}{%
					$\bestScore\gets\trialScore$  \;
					$a\gets r$ \;
				}
			}
			\If{$a \neq \emptyset$}{%
				$\done\gets\false$\;
				$K\gets K\cup\{a\}$\;
			}
		}
		\Return $K$ \;
	% }
\end{algorithm}

Note that the main purpose of adding randomly generated invariants is not to divide the models into more blocks in all cases, but to increase the robustness of the algorithm by the automatic discovery of important invariants in some cases.

%%%%%%%%%%%%%%%%%%%%%%%%%%%%%%%
\section{The Invariant Algorithm}
\label{sec:algorithm}
First, we describe the algorithm without the randomly generated invariants.
For each domain element of a model, we calculate the basic invariants and put them into an ordered list, which is called the invariant vector of that element.
If the model has multiple unary, binary, or ternary operations, then invariant vectors are calculated for each of them for all elements, and all the invariant vectors of the same domain element are concatenated to form a combined invariant vector for that domain element.
It follows that each model with $n$ domain elements will be associated with $n$ combined invariant vectors. %It should be obvious that 
Isomorphic models must have the same set of invariant vectors. 

To facilitate comparisons of invariant vectors, we sort the elements by their invariant vectors  lexicographically.  It follows that models isomorphic to each other must have the same sorted invariant vectors.

Our goal is not only to compare two models for isomorphism but to extract all non-isomorphic models from a list of models.  In that case, we set up a hash map to store blocks of the models where models in different blocks are guaranteed to be non-isomorphic. We use the invariant vectors for each model to send the model efficiently to the block (in the hash map) to which it belongs. That is, the keys in this hash map are the invariant vectors, and the values are the blocks of the models.  After all models are hashed into the hash map, the blocks stored in the hash map can be processed separately, and possibly in parallel, to extract one representative model from each isomorphism class.
An example of construction and use of invariants can be found in \cite{araujo_et_al:LIPIcs.CP.2021.4}.

To add random invariants to the algorithm, a preprocessing step is added to select an optimal subset of random invariants before the normal process. As described in Section~\ref{subsec:optimal:random:invariants}, we construct a list of random invariants, calculate basic invariants and the random invariants on a small sample of the input models. Then apply the greedy algorithm to find an optimal subset of random invariants for further processing (see Section~\ref{subsec:optimal:random:invariants}).  Finally, proceed to normal processing with the basic invariants and the optimal random invariants together.

Note that our invariant-based algorithm does not compare models for isomorphism.  It only cuts down the size of the problem to improve the performance of existing isomorphism filters such as Mace4's \textit{isofilter}.

Invariants have the potential to cope with the increasing size of the order of the algebra very well as illustrated in the following example.
Suppose an invariant with extremely low discriminating power gives only 2 values, 0 and 1. However, when applied to the models of an algebra of order 2, it could actually give 4 possible invariant vectors: [0,0], [0,1], [1,0], [1,1]. 
It is easy to generalize this observation: applying an invariant that gives at most $m$ values to the models of an algebra of order $n$ could result in a maximum of $m^n$ distinct invariant vectors. Furthermore, if $k$ invariants give \{$m_1, m_2,\cdots,m_k$\} values, then the maximum number of invariant vectors would be
\begin{equation}\label{eq:max:invariants}
\prod_{i=1}^k{m_i}^n 
\end{equation}
From the above analysis on the intricate interactions between invariants, we can make two more important observations:
\begin{enumerate}
	\item Combining invariants of low discriminating powers could give an invariant vector of surprisingly high discriminating power if they are targeting different areas of the algebraic structures.
	\item In general, the number of non-isomorphic models increases rapidly as the order of the algebra increases, but so does the maximum number of possible invariant vectors. This helps invariants to retain their discriminating powers to some extent as the order of the algebra increases. This explains why the invariant-based algorithm is very scalable.
\end{enumerate}

%%%%%%%%%%%%%%%%%%%%%%%%%%%%%%%

\section{Experimental Results}%
\label{sec:experiments}
We have implemented an invariant-based preprocessor to Mace4's isomorphic models filters.
We run experiments on a 6-core Intel\textregistered\,Core\textsuperscript{\texttrademark} i7-9850H CPU computer, with 32 Gb RAM installed.

The ALF database~\cite{alfAraujoMatosRamires} contains a collection of 158 algebras of high interest to the research community of algebra. Their definitions are conveniently given in first-order formulas that Mace4 can directly process. We use Mace4 to generate models for each algebra of the highest possible order that it can complete within 2 minutes. Mace4 is not able to generate models for 5 of them within that time limit, and they are excluded from the tests. The excluded algebras are: \#112 Kleene algebra, \#113 Concurrent Kleene algebra, \#114 Omega algebra, \#137 Steiner quasigroup, and \#138 Steiner loop.  
In addition, Mace4's \textit{isofilter} is not able to handle two of the largest algebras 
(\#8 BL-algebras and \#56 Linear Heyting algebras), each has between 1 to 3 million models.
These two algebras are also excluded from most of the statistics of the experimental results. So, we end up with 151 algebras in many of our analyses.

When random invariants are used in the experiment, the number of randomly generated invariants is 50,  
but at most 20 of the best of them will be used.  
The maximum depth of the expression tree (see Section~\ref{subsec:generationOfRandomInvariants}) is 4,
and the maximum number of variables in it is 3.
In this section, random invariants are used unless otherwise specified.

Since we run a large comprehensive set of test cases for comparison, 
the size of each test case is necessarily limited (uniformly and systematically to make comparisons of results meaningful) by the computing resources available.  
However, even for the small model sizes used in our experiments, the addition of the invariant-based algorithm improves the overall speed by an order of magnitude, without using parallel processing (see Table~\ref{table:invariants:time}). 

The overheads of calculating invariants are observed to be on average about 20 to 30\% of the total run time
in our experimental setting (see the third column in~\ref{table:invariants:time}).
However, invariants improve the speed by orders of magnitudes for big algebras. For example, for the longest (in terms of runtime) 10 algebras in our experiment, the invariant-based algorithm improves the overall speed by over 50 times (see Table~\ref{table:invariants:time}).
In fact, a very desirable feature of the invariant-based algorithm is that the improvement increases dramatically as the size of the set of models grows. 
Granted, for algebras with short runtime, the use of invariants may not pay off.
But for those cases, the degradation is really insignificant (see Fig.~\ref{fig:scatter_time}). Thus, in general, there is no need to have special logic to decide when not to use invariants.

Furthermore, Mace4's \textit{isofilter} is not able to handle two of the largest data sets, but our invariant-based algorithm can partition these models into smaller blocks to fit in Mace4's limits (see Table~\ref{table:invariants:time}).

\begin{table}[htbp]
	\caption{Isomorphism Filtering, w/ vs.\ w/o Invariants}
	\begin{tabular}{lrrrrr}
		\toprule
		& & & \multicolumn{2}{c}{Total Runtime (s)} \\ \cmidrule(r){4-5}
		\multicolumn{1}{l}{}  & \multicolumn{1}{p{1.25cm}}{\centering \#Mace4 Outputs} & \multicolumn{1}{p{1.9cm}}{\centering Invariant Calc. Time (s)}& \multicolumn{1}{p{1.4cm}}{\centering With Invariants} & \multicolumn{1}{p{1.4cm}}{\centering Without Invariants} \\ \cmidrule{1-5}
		Shortest 10 Algebras & 600 & 0.2 & 0.5 & 0.1  \\
		Longest 10 Algebras & 9,239,818 & 430 & 1,982 & 87,591  \\  
		All 151 Algebras & 33,643,548 & 1,500 & 5,030 & 95,952  \\ 
		2 Isofilter Failed Algebras & 4,075,054 & 208 & 727.3 & N/A  \\ \bottomrule
	\end{tabular}
	\label{table:invariants:time}
\end{table}

\begin{figure}[th]
	\centering
 	% File:  scatter_plot_time.tex
% Author: mikolas
% Created on: Mon May 24 15:28:36 CEST 2021
% Copyright (C) 2021, Mikolas Janota
  
\begin{tikzpicture}[scale=0.8, every node/.style={scale=0.8}]

  % \draw[color=lightgreen,fill=lightgreen] (0, 0) -- (1.95, 0) -- (5.43,3.48) -- (5.43,3.48) -- (5.43,5.43) -- (3.48,5.43) -- (0,1.95) -- (0,0);

  \begin{axis}[color=black!60,
    xmode=log,
    ymode=log,
    enlarge x limits=0,
    enlarge y limits=0,
    log ticks with fixed point,
%    axis equal,
    width=7cm,
    height=7cm,
    xmax=70000,
    ymax=70000,
    % for log axes, x filter operates on LOGS.
    % and log(x * 1000) = log(x) + log(1000):
%    x filter/.code=\pgfmathparse{#1 + 6.90775527898214},
]
    \addplot+[
    only marks,
%    scatter,
    mark=+,
    darkblue,
    mark size=3pt] table {scatter_time.data.tex};
  \end{axis}

  \draw[color=darkred] (0, 0) -- (5.4,5.4);
%  \draw (1.95, 0) -- (5.43,3.48);
  %  \draw (0,1.95) -- (3.48,5.43);

  % \draw[color=lightgray,dotted] (0,4.835) -- (5.4,4.835);
  % \draw[color=lightgray,dotted] (4.835,0) -- (4.835,5.4);

  \node at (2.7,-.5) {w/ invariants};
  \node[rotate=90] at (-.9, 2.7) {w/o invariants};

  \node[below left] at (5.5,-0.2) {{\footnotesize time (s)}};

  \node[left] at (0, 5.3) {{\footnotesize time}};
  \node[left] at (0, 5) {{\footnotesize (s)}};

\end{tikzpicture}
 	\caption{Runtimes: w/ vs.\ w/o Invariants (151 ALF Algebras)}\label{fig:scatter_time}
\end{figure}
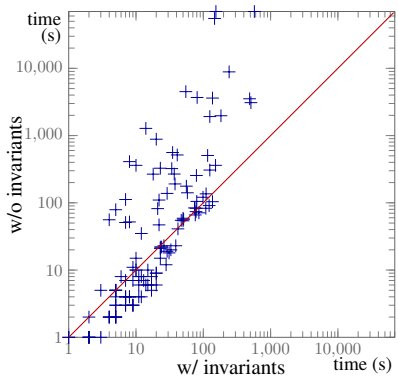

The performance of the invariant-based algorithm relies heavily on the discriminating power of its invariants. 
The best possible case is that only 1 non-isomorphic model is in every block, in which case, only $m-1$ comparisons of models are needed to eliminate all isomorphic models from a block of $m$ models. 
Our invariants are quite powerful as evidenced by the fact that the average number of non-isomorphic models per block is very close to 1 for the 151 algebras in the experiment (see Table~\ref{table:invariants:percentile}).

\begin{table}[htbp]
	\caption{Discriminating Power (153 ALF Algebras)}
	\begin{tabular}{crr}
		\toprule
		 & \multicolumn{2}{c}{\centering Avg \#Non-isomorphic Models per Block} \\ \cmidrule(r){2-3}
	    \multicolumn{1}{c}{\centering Percentile} & \multicolumn{1}{c}{\centering w/ Random Invariants}  & \multicolumn{1}{c}{\centering w/o Random Invariants} \\
		\cmidrule{1-3}
		95th & 1.346 & 2.677 \\  
		80th & 1.036 & 1.179 \\ 
		60th & 1.003 & 1.018 \\
		40th & 1.000 & 1.005 \\ \bottomrule
	\end{tabular}
	\label{table:invariants:percentile}
\end{table}

% \begin{table}[htbp]
% 	\caption{\# of Random Invariants (153 ALF Algebras)}
% 	\begin{tabular}{cr}
% 		\toprule
% 		\multicolumn{1}{c}{\centering Percentile} & \multicolumn{1}{c}{\centering \#Random Invariants} \\
% 		\cmidrule{1-2}
% 		95th & 8 \\  
% 		80th & 5 \\ 
% 		60th & 3 \\
% 		40th & 2 \\ \bottomrule
% 	\end{tabular}
% 	\label{table:invariants:percentile:number}
% \end{table}

\begin{figure}[th]
\centering
\begin{tikzpicture}
	\begin{axis}[
		title={\# Random Invariants Used},
		xlabel={\# Random Invariants},
		ylabel={Cumulative Frequency (\%)},
		xmin=0, xmax=19,
		ymin=0, ymax=100,
		xtick={0,1,2,3,4,5,6,7,8,9,10,11,12,13,14,15,16,17,18,19},
		ytick={10,20,40,60,80,90,100},
		legend pos=south east,
		ymajorgrids=true,
		grid style=dashed,
		]
		
		\addplot[
		color=blue,
		mark=square,
		]
		coordinates {
			(0,10)(1,22)(2,46)(3,61)(4,72)(5,82)(6,91)(7,94)(8,97)(9,98)
			(10,98)(11,98)(12,98)(13,99)(14,99)(15,99)(16,99)(17,99)(18,99)(19,100)
		};
		\legend{Cum. Frequency}
		
	\end{axis}
\end{tikzpicture}
\caption{\# of Random Invariants (153 ALF Algebras)}\label{fig:cum:random:invariants}
\end{figure}
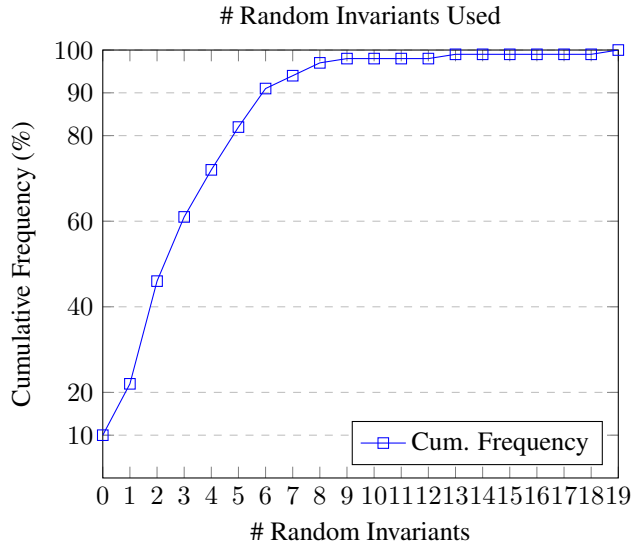

\subsection{Basic Invariants vs.\ Basic Invariants + Random Invariants}
\label{subsec:basicVsRandom}
As shown in Table~\ref{table:invariants:percentile}, the hand-crafted basic invariants have  very good discriminating power (see the last column in the table). 
Nevertheless, the addition of random invariants improves the discriminating powers (see the middle column of the table).
This increase in discriminating power comes with a small overhead in processing time as
the number of random invariants is quite small, usually just a few. For example, 6 or fewer random invariants are used in about 90\% of the algebras (see Figure~\ref{fig:cum:random:invariants}). 
For the case when the basic invariants are already doing a very good job, the addition of random invariants may not pay off.  But the degradation is minimal because the job would finish fast when the discriminating powers of the invariants are high (See the scatter plot Fig.~\ref{fig:scatter_time_basic_vs_random}). Therefore, there is no need for special logic to decide when not to use random invariants.
In our experiment, the overall run time for all 151 algebras is reduced when random invariants are added, with most of the improvements coming from the top 3 algebras, which are among the algebras that take the longest to finish (see Table~\ref{table:invariants:random:time} and Fig.~\ref{fig:scatter_time_basic_vs_random}).

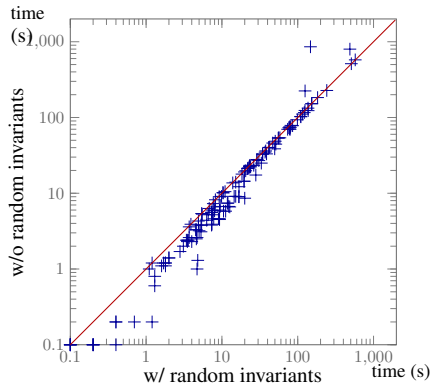
\begin{figure}[th]
	\centering
	% File:  scatter_plot_time.tex
% Author: mikolas
% Created on: Mon May 24 15:28:36 CEST 2021
% Copyright (C) 2021, Mikolas Janota
\begin{tikzpicture}[scale=0.8, every node/.style={scale=0.8}]

  % \draw[color=lightgreen,fill=lightgreen] (0, 0) -- (1.95, 0) -- (5.43,3.48) -- (5.43,3.48) -- (5.43,5.43) -- (3.48,5.43) -- (0,1.95) -- (0,0);

  \begin{axis}[color=black!60,
    xmode=log,
    ymode=log,
    enlarge x limits=0,
    enlarge y limits=0,
    log ticks with fixed point,
%    axis equal,
    width=7cm,
    height=7cm,
    xmax=2000,
    ymax=2000,
    % for log axes, x filter operates on LOGS.
    % and log(x * 1000) = log(x) + log(1000):
%    x filter/.code=\pgfmathparse{#1 + 6.90775527898214},
]
    \addplot+[
    only marks,
%    scatter,
    mark=+,
    darkblue,
    mark size=3pt] table {scatter_time_basic_vs_random.data.tex};
  \end{axis}

  \draw[color=darkred] (0, 0) -- (5.4,5.4);
%  \draw (1.95, 0) -- (5.43,3.48);
  %  \draw (0,1.95) -- (3.48,5.43);

  % \draw[color=lightgray,dotted] (0,4.835) -- (5.4,4.835);
  % \draw[color=lightgray,dotted] (4.835,0) -- (4.835,5.4);

  \node at (2.7,-.5) {w/ random invariants};
  \node[rotate=90] at (-.9, 2.7) {w/o random invariants};

  \node[below] at (5.5,-0.2) {{\footnotesize time (s)}};

  \node[left,align=left] at (-0.3, 5.3) {\footnotesize time\\(s)};
  % \node[left] at (0, 5) {{\footnotesize (s)}};

\end{tikzpicture}
	\caption{Runtimes: w/ vs.\ w/o Random Invariants}
	\label{fig:scatter_time_basic_vs_random}
\end{figure}

\begin{table}[htbp]
	\caption{Isomorphism Filtering Time, w/ and w/o Random Invariants}
	\begin{tabular}{lrrrrr}
		\toprule
		& \multicolumn{2}{c}{Total Time (s)} \\ \cmidrule(r){2-3}
		\multicolumn{1}{l}{}  &  \multicolumn{1}{c}{\centering With Random Invariants} & \multicolumn{1}{c}{\centering Without Random Invariants} \\ \cmidrule{1-3}
		Top 3 Improved Algebras & 760 & 1,882  \\ 
		All 151 Algebras & 5,758 & 6,525  \\ \bottomrule
	\end{tabular}
	\label{table:invariants:random:time}
\end{table}

\subsection{Larger Data Sets}
\label{subsec:largerDataSets}
As remarked earlier, to cover a large variety of algebras, we have to limit the order of each algebra in the experiments. Small datasets do not adequately show the true advantage of the invariant-based algorithm. 
Here we present three examples in which we go two orders higher in each of the algebras, giving us test datasets of over a hundred million models. The first algebra (\#88 Quasi-MV-algebra) is defined by one binary operation and 1 relation, the second one (\#15 Brouwerian semilattices) by two binary operations, and the third one (\#4 BCK-join-semilattice) by two binary operations and 1 relation. As shown in Table~\ref{table:higherInvariants:time}, at the baseline when the order of the algebra is small, the invariant algorithm slows down the process very slightly.  However, as the order of the algebra goes higher, the invariant-based algorithm improves the speed by orders of magnitudes.
In some cases, Mace4 is not able to handle a large number of models. In all the tests, we also observe that the discriminatory powers of the invariants hold up quite well in large datasets (see Table~\ref{table:higherInvariants:avgModels}).

\begin{table}[htbp]
	\caption{Isomorphism Filtering, w/ vs.\ w/o Invariants for Higher Orders}
	\begin{threeparttable}
	\begin{tabular}{lrrrr}
		\toprule
		& & & \multicolumn{2}{c}{Total Runtime (s)} \\ \cmidrule(lr){4-5}
		& \multicolumn{1}{l}{Order}  & \multicolumn{1}{p{1.55cm}}{\centering \#Mace4 Outputs} & \multicolumn{1}{p{1.4cm}}{\centering With Invariants} & \multicolumn{1}{p{1.4cm}}{\centering Without Invariants} \\ \cmidrule{1-5}
		\#88 Quasi-MV-algebra & 7 & 10,902 & 1.6 & 0.7  \\
		& 8 & 4,793,924 &  558 & 30,701  \\ 
		& 9 & 29,799,618 & 3,666 & N/A\tnote{1}  \\ 
		\#15 Brouwerian semilattices & 6 & 47,349 & 12 & 4  \\
		& 7 & 2,247,564 & 440 & 1,964  \\ 
		& 8 & 146,875,177 & 40,017 & N/A  \\ 
		\#4 BCK-join-semilattice & 9 & 122,754 & 23 & 15  \\
		& 10 & 1,175,784 & 305 & 532  \\ 
		& 11 & 12,307,002 & 1,213 & 54,524  \\ \bottomrule
	\end{tabular}
	\begin{tablenotes}\footnotesize
	\item[1] Isofilter fails after processing 4.5 million (15\% of all) models in 11.5 hours.
	\end{tablenotes}
	\end{threeparttable}
	\label{table:higherInvariants:time}
\end{table}

\begin{table}[htbp]
\caption{Discriminating Power of Invariants for Higher Orders}
	\begin{tabular}{lrrrr}
		\toprule
		& & & \multicolumn{2}{c}{Non-isomorphic Models} \\ \cmidrule(lr){4-5}
		& \multicolumn{1}{l}{Order}  & \multicolumn{1}{p{1.7cm}}{\centering \#Blocks} & \multicolumn{1}{p{1.4cm}}{\centering Total} & \multicolumn{1}{p{1.9cm}}{\centering Avg per Block} \\ \cmidrule{1-5}
		\#88 Quasi-MV-algebra & 7 & 567 & 477 & 1.19  \\
		& 8 & 153,163 & 55,544 & 2.76  \\ 
		& 9 & 264,972 & 141,750 & 1.87   \\ 
		\#15 Brouwerian semilattices & 6 & 745 & 745 & 1.00  \\
		& 7 & 8,272 & 8,272 & 1.00  \\ 
		& 8 & 115,801 & 114,943 & 1.01  \\ 
		\#4 BCK-join-semilattice & 9 & 26 & 26 & 1.00  \\
		& 10 & 47 & 47 & 1.00  \\ 
		& 11 & 82 & 82 & 1.00  \\ \bottomrule
	\end{tabular}
\label{table:higherInvariants:avgModels}
\end{table}

\subsection{Parallel Processing}
\label{subsec:parrallel:processing}
The invariant-based algorithm is very scalable since the data are divided into blocks that can be processed independently as long as resources are available. In this experiment, we apply parallel processing to the top 3 algebras with the longest runtimes (close to 500s or more). We run each of them with 5 parallel threads and see about a 50 - 60\% reduction in run times (see Table~\ref{table:invariants:parallel:time}).  The results would be even better if more resources are available as a large number of blocks are available in these cases.
\begin{table}[htbp]
	\caption{Isomorphism Filtering, Serial vs.\ Parallel}
	\begin{tabular}{lrrrr}
		\toprule
		& & & \multicolumn{2}{c}{Runtime (s)} \\ \cmidrule(r){4-5}
		\multicolumn{1}{l}{} & {\centering Order} &
		{\centering \#Blocks} &
		 \multicolumn{1}{c}{\centering Serial} & \multicolumn{1}{c}{\centering Parallel} \\ \cmidrule{1-5}
		\#8 BL-algebras & 5 & 735,820 & 574 & 254  \\  
		\#20 Commutative lattice-ordered monoids & 7 & 15,499 & 510 & 240 \\
		\#145 Digroup & 12 & 17 & 488 & 177 \\  \bottomrule
	\end{tabular}
	\label{table:invariants:parallel:time}
\end{table}

%%%%%%%%%%%%%%%%%%%%%%%%%%%%%%%

\section{Related Work}%
\label{sec:relatedwork}
Classes of algebras can often be defined in first-order formulas. For these algebras, a finite model finder, such as Mace4~\cite{Mccune_2003_tech264}, SEM~\cite{DBLP:conf/ijcai/ZhangZ95}, and FALCON~\cite{ZhangJ1996Falcon}, FMSET~\cite{DBLP:journals/fuin/BenhamouH99}, etc., can work on finding all their models. 
A well-known issue with this approach is that first-order formulas introduce symmetries into
the problem, which leads to the generation of a huge number of isomorphic models in the outputs~\cite{DBLP:conf/frocos/RegerR019}. These isomorphic models can either be suppressed in the
search phase or be removed in a postprocessing step after the models are generated. 

Past work in enumerating non-isomorphic finite models has focused on not generating isomorphic models in the search phase. 
Symmetry breaking is thus a central focus of their research~\cite{Claessen03newtechniques,DBLP:conf/frocos/RegerR019,articleCrawfordGinsbergLuksRoy1997}.
An excellent example of a simple, powerful, and general algorithm to break symmetry \textit{dynamically} is the least number heuristic (LNH)~\cite{ZhangJ1996Falcon,DBLP:journals/jar/AudemardBH06}, which picks the smallest one not used so far when a new domain element is to be selected during the search.
This is implemented in many solvers such as Mace4, SEM, FMSET, and FALCON\@. Another symmetry breaker, the eXtended least number heuristic (XLNH)~\cite{DBLP:journals/jar/AudemardBH06,Audemard2001XLNH}, is based on similar ideas as the LNH, but could give better performance if there is at least one unary operation in the FOL clauses that define the model. It is also implemented in many finite model enumerators such as SEM\@. 

The underlying idea of LNH can also be applied to break symmetries \textit{statically}.
This is necessary for approaches where we do not wish to modify the underlying solver.
This is the case for finite model finders based on SAT solvers~\cite{DBLP:conf/frocos/RegerR019,Claessen03newtechniques,DBLP:conf/lpar/JanotaS18}.
The issue is the overhead of encoding LNH in conjunctive normal form (CNF)  as well as its complex interaction with the SAT solver.
Originally, LNH  was only encoded for constants~\cite{Claessen03newtechniques}.
Later, with additional effort,  it was shown that other terms can also be considered~\cite{DBLP:conf/frocos/RegerR019}.
 
% For example, if there are multiple constants \{$c_0, c_1, \cdots$\} in the
% inputs, then the first constant can simply be fixed to be 0 by adding the input
% clauses $c_0=0$, the next constant can be fixed to be 0 or 1 by the input
% clause $c_1=0 \vee c_1=1$, and so on. The disadvantage of this static symmetry
% breaking method that it increases the size, and hence the processing time, of
% the inputs. The advantage is that no source code change is necessary for
% systems that do not already have LNH built into them.

The addition of symmetry-breaking input clauses could be useful in steering the searcher away from  the needless exploration of sub-search space~\cite{articleCrawfordGinsbergLuksRoy1997}. 
For example, it is well-known that a finite semigroup has at least one idempotent element, 
so we may add the clause 0 * 0 = 0 to the list of input clauses to cut off the search of the branch 0 * 0 = 1, etc. 
However, this kind of symmetry breaking often requires deep knowledge of the algebra in question, which
may not be available when the algebra is first studied.

Most importantly, these symmetry-breaking techniques do not guarantee isomorph-freeness. While not generating isomorphic models would be ideal, but to guarantee that isomorphic models are not produced in the search phase is a hard problem that few modern-day solvers attempt to do. 
Systems that do try to do so, 
such as SEMK~\cite{BOYDELATOUR200591,MCKAY1998306} and SEMD~\cite{ZhangJia2006Iso}, are either yet-to-be-completed or are better off allowing some isomorphic models in the outputs for some cases for better efficiency.

When isomorphic models are not totally suppressed in the search phase, they would need
to be removed in the post-processing steps. Many of them use a limited number of invariants to help speed up the process.
Mace4, for example, has a program, \textit{isofilter}, to filter out 
isomorphic models that it generates in the search phase. It calculates one invariant, frequency of occurrence of domain element, that is, the number of times a domain element appears in the operation tables.
It uses this invariant to help separate non-isomorphic models and to help guide the construction
of isomorphic functions between potentially isomorphic models. 
Needless to say, the discriminating power of one single invariant is limited.  
Indeed, it fails miserably when the operation table is a Latin square, which is the case for quasigroups.
To alleviate this issue, Mace4 has another program, \textit{isofilter2}, which does not try to construct 
isomorphic functions between models, but to convert models to their canonical forms based on the same algorithm~\cite{MCKAY1998306} used by SEMK\@.  \textit{Isofilter2} works very well with quasigroups, but
the overhead in computing the canonical forms of the models is so high that it becomes slower than \textit{isofilter} for many algebraic structures such as semigroups. 

The loops package~\cite{loops3.4.1} in GAP~\cite{GAP4} is not a finite model enumerator 
but provides a stand-alone function to extract non-isomorphic models from a list of quasigroups.  
It uses 9 invariants, some of which are expensive to calculate, to help separate non-isomorphic
models, and to guide the construction of isomorphic functions between models having the same invariant vectors. These 9 invariants exploit the specific properties of the quasigroup, and may not be effective
for other algebraic structures.  On the other hand, our invariants target many different areas in the common algebraic structures. Furthermore, their invariant vectors are limited to one binary operation table. Our invariant vectors can be constructed from multiple unary, binary, and ternary operation tables. 

Invariants are sometimes incorporated into the finite-model enumerating algorithm for specific algebras. 
These invariants are sometimes very simple and easy to compute, such as those in the algorithm for enumerating quandles~\cite{ElhandadiMacquarrieRestrepoAutoGroupsQuandles}.
But others may be very complicated, not easy to implement, and not cheap to compute as 
in the case of the algorithm for enumerating inverse semigroups~\cite{Malandro2019EnumerationOF} using the constraint solver, Minion~\cite{DBLP:conf/ecai/GentJM06}. Furthermore, the number of invariants used in these cases is usually very small, often two or fewer. Recall expression~\ref{eq:max:invariants} on page \pageref{eq:max:invariants} which shows that the number of invariants could increase the discriminating power drastically.
Our invariants are cheap to calculate, are applicable to more algebraic structures, and are high in number to provide more opportunities for separating non-isomorphic models.   Moreover, they can easily be incorporated into any finite model enumerator.

Neither Mace4's nor loops' isomorphic model
filters make use of the hash table to store the models so that non-isomorphic models will 
never be compared once they are separated by their invariant vectors.
This introduces some inefficiencies in their algorithms. Thus, both could benefit immensely from the reduced number of models in the blocks created by our invariant-based algorithm as a preprocessing step.

Another important feature in the invariant-based algorithm is randomization.
Using randomization in the search phase in Boolean Satisfiability (SAT) and 
Constrained Programming (CSP) algorithms is a tried and tested technique~\cite{DBLP:conf/aaai/GomesSK98,DBLP:conf/cp/BaptistaS00,Lynce02completeunrestricted}. 
This strategy is built into many SAT solvers such as Chaff~\cite{DBLP:conf/dac/MoskewiczMZZM01}. 
However, using randomly generated invariants to help separate non-isomorphic models in the finite model enumeration is a novel idea.
It helps solve the hardest problems in filtering isomorphic models as shown in our experiments,
and consequently, increases the robustness of the invariant-based algorithm.

%%%%%%%%%%%%%%%%%%%%%%%%%%%%%%%

\section{Future Work and Conclusions}%
\label{sec:futurework}

As pointed out in section~\ref{sec:experiments}, the efficiency of the invariant-based algorithm relies
heavily on the discriminating powers of both the hand-crafted and the randomly generated invariants. 
Future work will therefore concentrate on finding powerful invariants to target common algebraic structures,
and to find the best parameters to generate optimal random invariants. 
For example, what depth and breadth of the expression tree would be most cost-effective in
generating random invariants?  
What is the best range of ratios of binary operations to unary operations in a random invariant?
What is the best size of the sample to use in finding the optimal set of random invariants?

In summary, we present in this paper an algorithm that uses invariants both as discriminators and as hash keys to partition a set of models into blocks, in which no models across blocks are isomorphic. The blocks are
hashed into a hasp map so that they will not be processed together.
Included in the algorithm is
the novel idea of using randomly generated invariants to supplement hand-crafted invariants to make the algorithm more robust. 
We show that the invariant-based algorithm is simple,
efficient, scalable, and parallelizable. It is also compatible with most, if not all, existing finite model enumerators. 
It can be used as a stand-alone preprocessor to split models into blocks to feed
into isomorphic model filters, or it can be directly incorporated into them.
Future research will concentrate on finding powerful invariants in different areas of algebraic structures, and on the automatic discovery of optimal random invariants.

%%%%%%%%%%%%%%%%%%%%%%%%%%%%%%%
\section*{Acknowledgments}
The results were supported by the Ministry of Education, Youth and Sports within the dedicated program ERC CZ under the project POSTMAN no.~LL1902 and
Funda{\c c}\~ao para a Ci\^encia e a Tecnologia,
% (Portuguese Foundation for Science and Technology)
through the projects UIDB/00297{-}/2020
% (Centro de Matem\'atica e Aplica{\c c}\~oes),
(CMA),
PTDC/MAT-PUR/31174/2017, UIDB/04621/2020 and UIDP/04621/2020.
This scientific article is part of the RICAIP project that has received funding from the European
Union's Horizon~2020 research and innovation programme under grant agreement No~857306.

\bibliographystyle{plain}
\bibliography{refs}
\end{document}